\documentclass[letterpaper, 10 pt, conference]{ieeeconf}   

\IEEEoverridecommandlockouts                              
\usepackage{amsmath} 
\usepackage{amssymb}  
\usepackage{comment}
\usepackage{graphicx}
\usepackage{tikz}

\usepackage{cite}
\usetikzlibrary{arrows.meta, positioning, shapes.geometric}
\usepackage{hyperref}

\usepackage{booktabs}

\newtheorem{definition}{Definition}[section]
\newtheorem{theorem}{Theorem}[section]

\newtheorem{remark}{Remark}[theorem]
\newtheorem{lemma}[theorem]{Lemma}
\newtheorem{proposition}[theorem]{Proposition}
\newtheorem{assumption}[theorem]{Assumption}

\DeclareMathOperator*{\argmin}{arg\,min}

\usepackage[prependcaption,colorinlistoftodos]{todonotes}

\title{\LARGE \bf
 Control Strategies for Recommendation Systems in Social Networks
}

\author{Ben Sprenger, Giulia De Pasquale, Raffaele Soloperto, John Lygeros, and Florian D\"orfler
\thanks{The authors are with the Automatic Control Laboratory, Department of Electrical Engineering and Information Technology, ETH Zurich, Physikstrasse
3 8092 Zurich, Switzerland. (e-mail:\{bsprenger, degiulia, soloperr, jlygeros,
dorfler\}@ethz.ch). This research is supported by the Swiss National Science Foundation under NCCR automation.
}}

\begin{document}

\maketitle
\thispagestyle{empty}
\pagestyle{empty}

\begin{abstract}
A closed-loop control model to analyze the impact of recommendation systems on opinion dynamics within social networks is introduced. The core contribution is the development and formalization of model-free and model-based approaches to recommendation system design, integrating the dynamics of social interactions within networks via an extension of the Friedkin-Johnsen (FJ) model. Comparative analysis and numerical simulations demonstrate the effectiveness of the proposed control strategies in maximizing user engagement and their potential for influencing opinion formation processes.

\end{abstract}


\section{Introduction}

\emph{Motivation:}
Opinion dynamics models provide insights into how opinions form and spread within social networks \cite{AVP-RT:17}. These models explore the interplay between individuals' beliefs, their interactions with others, and the emergence of phenomena such as consensus or polarization \cite{CA:12,AF-MM-TF-ECB-GD-SH-JL:17,AVP-RT:17,CA:23}. The rise of recommendation systems has added a new dimension to the study of opinion dynamics. These systems are responsible for filtering massive amounts of available content and providing tailored suggestions to users to ensure engaging experiences on online platforms. While recommendation systems offer benefits in terms of information filtering, they can fundamentally alter the dynamics of opinion formation \cite{DG-JL-PH:19,AS-DP-FG-JK:19,NP-JB-EE-GDP-SB-AA:23,NL-FD-NP:23}. Understanding the interplay between recommendation systems and opinion dynamics is essential for ensuring that personalized content delivery aligns with desirable social outcomes such as fairness and diversity.

\emph{Related work:}
The work in~\cite{WSR-JWP-PF:22} introduces the feedback interaction between a recommendation system and a single user. The recommendation system provides binary opinions to enhance engagement maximization based on the user's clicking history.  
The work shows that, when engagement maximization and confirmation bias are the only mechanisms involved, recommendation systems induce polarization over users' opinions. Importantly, their work does not consider the influence that interacting users' opinions might have.  Social influences are incorporated into a recommendation system in~\cite{JC-JL-GZ-YD-LM:18} via the classic DeGroot opinion dynamics model. They show that, by taking into account opinion group dynamics, the quality of recommendations made to users improves. However, the effect of the recommendation system on the opinion formation process itself is not examined. Further, the DeGroot opinion dynamics model might prove too simplistic for online social platform, where complex phenomena such as polarization and echo-chamber formation need to be understood.   
 Given the interpretation of recommendation systems as inputs to a dynamical system, the work in~\cite{MG-DC-NK-DM:19:cdc} represents an important foundation for control-theoretic analyses of opinion formation for the Friedkin-Johnsen (FJ) model.  For an overview of the properties and variations of the FJ model we refer the reader to~\cite{BDOA-MY:19}.

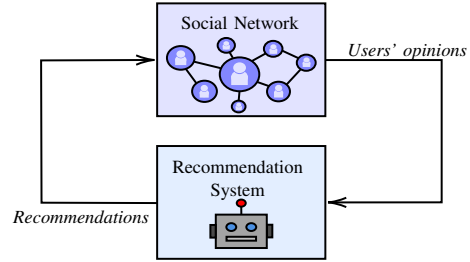
\begin{figure}[]
    \centering

\tikzset{every picture/.style={line width=0.75pt}} 
\begin{tikzpicture}[x=0.7pt,y=0.6pt,yscale=-0.9,xscale=0.9]

\draw  [fill={rgb, 255:red, 226; green, 227; blue, 255 }  ,fill opacity=1 ] (179.67,30.4) -- (280.67,30.4) -- (280.67,109.47) -- (179.67,109.47) -- cycle ;
\draw    (179.75,170.38) -- (110.25,170.13) -- (110,70.38) -- (176,70.08) ;
\draw [shift={(178,70.07)}, rotate = 179.74] [color={rgb, 255:red, 0; green, 0; blue, 0 }  ][line width=0.75]    (10.93,-3.29) .. controls (6.95,-1.4) and (3.31,-0.3) .. (0,0) .. controls (3.31,0.3) and (6.95,1.4) .. (10.93,3.29)   ;
\draw  [fill={rgb, 255:red, 131; green, 135; blue, 255 }  ,fill opacity=1 ] (217.27,80.13) .. controls (217.27,73.58) and (222.58,68.27) .. (229.13,68.27) .. controls (235.69,68.27) and (241,73.58) .. (241,80.13) .. controls (241,86.69) and (235.69,92) .. (229.13,92) .. controls (222.58,92) and (217.27,86.69) .. (217.27,80.13) -- cycle ;
\draw  [fill={rgb, 255:red, 131; green, 135; blue, 255 }  ,fill opacity=1 ] (243.77,62.26) .. controls (243.77,59.18) and (246.26,56.68) .. (249.34,56.68) .. controls (252.42,56.68) and (254.92,59.18) .. (254.92,62.26) .. controls (254.92,65.34) and (252.42,67.83) .. (249.34,67.83) .. controls (246.26,67.83) and (243.77,65.34) .. (243.77,62.26) -- cycle ;
\draw  [fill={rgb, 255:red, 131; green, 135; blue, 255 }  ,fill opacity=1 ] (245.65,92.2) .. controls (245.65,88.4) and (248.73,85.32) .. (252.54,85.32) .. controls (256.34,85.32) and (259.42,88.4) .. (259.42,92.2) .. controls (259.42,96.01) and (256.34,99.09) .. (252.54,99.09) .. controls (248.73,99.09) and (245.65,96.01) .. (245.65,92.2) -- cycle ;
\draw  [fill={rgb, 255:red, 131; green, 135; blue, 255 }  ,fill opacity=1 ] (224.5,102.71) .. controls (224.5,100.44) and (226.35,98.6) .. (228.62,98.6) .. controls (230.89,98.6) and (232.73,100.44) .. (232.73,102.71) .. controls (232.73,104.98) and (230.89,106.82) .. (228.62,106.82) .. controls (226.35,106.82) and (224.5,104.98) .. (224.5,102.71) -- cycle ;
\draw  [fill={rgb, 255:red, 131; green, 135; blue, 255 }  ,fill opacity=1 ] (200.66,92.04) .. controls (200.66,87.93) and (203.99,84.6) .. (208.1,84.6) .. controls (212.22,84.6) and (215.55,87.93) .. (215.55,92.04) .. controls (215.55,96.16) and (212.22,99.49) .. (208.1,99.49) .. controls (203.99,99.49) and (200.66,96.16) .. (200.66,92.04) -- cycle ;
\draw  [fill={rgb, 255:red, 131; green, 135; blue, 255 }  ,fill opacity=1 ] (216,59.82) .. controls (216,56.96) and (218.32,54.64) .. (221.18,54.64) .. controls (224.04,54.64) and (226.36,56.96) .. (226.36,59.82) .. controls (226.36,62.69) and (224.04,65.01) .. (221.18,65.01) .. controls (218.32,65.01) and (216,62.69) .. (216,59.82) -- cycle ;
\draw  [fill={rgb, 255:red, 131; green, 135; blue, 255 }  ,fill opacity=1 ] (185.56,68.62) .. controls (185.56,64.04) and (189.27,60.33) .. (193.84,60.33) .. controls (198.42,60.33) and (202.13,64.04) .. (202.13,68.62) .. controls (202.13,73.2) and (198.42,76.91) .. (193.84,76.91) .. controls (189.27,76.91) and (185.56,73.2) .. (185.56,68.62) -- cycle ;
\draw  [fill={rgb, 255:red, 131; green, 135; blue, 255 }  ,fill opacity=1 ] (263.52,72) .. controls (263.52,68.87) and (266.06,66.33) .. (269.19,66.33) .. controls (272.32,66.33) and (274.86,68.87) .. (274.86,72) .. controls (274.86,75.13) and (272.32,77.67) .. (269.19,77.67) .. controls (266.06,77.67) and (263.52,75.13) .. (263.52,72) -- cycle ;
\draw    (229.13,92) -- (228.62,98.6) ;
\draw    (239.64,85.45) -- (246.73,88.36) ;
\draw    (236.73,70.73) -- (244.91,65.45) ;
\draw    (254.91,63.64) -- (264.36,68.73) ;
\draw    (257.27,87.45) -- (265.64,76.91) ;
\draw    (223.27,64.91) -- (225.09,68.55) ;
\draw    (201.82,71.27) -- (217.82,75.82) ;
\draw    (196.91,76.36) -- (203.09,86.18) ;
\draw  [color={rgb, 255:red, 226; green, 227; blue, 255 }  ,draw opacity=1 ][fill={rgb, 255:red, 226; green, 227; blue, 255 }  ,fill opacity=1 ] (224.7,82.14) .. controls (224.7,79.82) and (226.68,77.94) .. (229.14,77.94) .. controls (231.59,77.94) and (233.58,79.82) .. (233.58,82.14) .. controls (233.58,84.46) and (231.59,86.34) .. (229.14,86.34) .. controls (226.68,86.34) and (224.7,84.46) .. (224.7,82.14) -- cycle ;
\draw  [color={rgb, 255:red, 226; green, 227; blue, 255 }  ,draw opacity=1 ][fill={rgb, 255:red, 226; green, 227; blue, 255 }  ,fill opacity=1 ] (224.73,81.92) -- (233.51,81.92) -- (233.51,86.96) -- (224.73,86.96) -- cycle ;
\draw  [color={rgb, 255:red, 131; green, 135; blue, 255 }  ,draw opacity=1 ][fill={rgb, 255:red, 131; green, 135; blue, 255 }  ,fill opacity=1 ] (225.94,76.22) .. controls (225.94,74.54) and (227.39,73.17) .. (229.17,73.17) .. controls (230.95,73.17) and (232.4,74.54) .. (232.4,76.22) .. controls (232.4,77.91) and (230.95,79.28) .. (229.17,79.28) .. controls (227.39,79.28) and (225.94,77.91) .. (225.94,76.22) -- cycle ;
\draw  [color={rgb, 255:red, 226; green, 227; blue, 255 }  ,draw opacity=1 ][fill={rgb, 255:red, 226; green, 227; blue, 255 }  ,fill opacity=1 ] (225.94,75.5) .. controls (225.94,73.82) and (227.39,72.45) .. (229.17,72.45) .. controls (230.95,72.45) and (232.4,73.82) .. (232.4,75.5) .. controls (232.4,77.19) and (230.95,78.56) .. (229.17,78.56) .. controls (227.39,78.56) and (225.94,77.19) .. (225.94,75.5) -- cycle ;

\draw  [fill={rgb, 255:red, 227; green, 238; blue, 255 }  ,fill opacity=1 ] (179.67,130.4) -- (280.67,130.4) -- (280.67,209.47) -- (179.67,209.47) -- cycle ;
\draw  [fill={rgb, 255:red, 164; green, 164; blue, 164 }  ,fill opacity=1 ] (214.88,178.49) -- (245.5,178.49) -- (245.5,201.9) -- (214.88,201.9) -- cycle ;
\draw  [fill={rgb, 255:red, 74; green, 144; blue, 226 }  ,fill opacity=1 ] (221.23,187.21) .. controls (221.23,185.71) and (222.56,184.48) .. (224.19,184.48) .. controls (225.82,184.48) and (227.15,185.71) .. (227.15,187.21) .. controls (227.15,188.72) and (225.82,189.95) .. (224.19,189.95) .. controls (222.56,189.95) and (221.23,188.72) .. (221.23,187.21) -- cycle ;
\draw  [fill={rgb, 255:red, 74; green, 144; blue, 226 }  ,fill opacity=1 ] (233.23,187.21) .. controls (233.23,185.71) and (234.56,184.48) .. (236.19,184.48) .. controls (237.82,184.48) and (239.15,185.71) .. (239.15,187.21) .. controls (239.15,188.72) and (237.82,189.95) .. (236.19,189.95) .. controls (234.56,189.95) and (233.23,188.72) .. (233.23,187.21) -- cycle ;
\draw  [fill={rgb, 255:red, 74; green, 74; blue, 74 }  ,fill opacity=1 ] (245.64,185.57) -- (248.83,185.57) -- (248.83,194.45) -- (245.64,194.45) -- cycle ;
\draw  [fill={rgb, 255:red, 255; green, 0; blue, 0 }  ,fill opacity=1 ] (227.47,170.42) .. controls (227.47,169.05) and (228.68,167.93) .. (230.17,167.93) .. controls (231.65,167.93) and (232.86,169.05) .. (232.86,170.42) .. controls (232.86,171.79) and (231.65,172.91) .. (230.17,172.91) .. controls (228.68,172.91) and (227.47,171.79) .. (227.47,170.42) -- cycle ;
\draw    (230.17,172.91) -- (230.22,178.72) ;
\draw  [fill={rgb, 255:red, 74; green, 74; blue, 74 }  ,fill opacity=1 ] (211.64,185.74) -- (214.83,185.74) -- (214.83,194.61) -- (211.64,194.61) -- cycle ;

\draw  [fill={rgb, 255:red, 191; green, 191; blue, 191 }  ,fill opacity=1 ] (221.32,194.02) -- (239.06,194.02) -- (239.06,196.75) -- (221.32,196.75) -- cycle ;

\draw    (280.75,70.13) -- (350.5,70.13) -- (350.5,169.88) -- (284.29,169.93) ;
\draw [shift={(282.29,169.93)}, rotate = 359.96] [color={rgb, 255:red, 0; green, 0; blue, 0 }  ][line width=0.75]    (10.93,-3.29) .. controls (6.95,-1.4) and (3.31,-0.3) .. (0,0) .. controls (3.31,0.3) and (6.95,1.4) .. (10.93,3.29)   ;
\draw  [color={rgb, 255:red, 226; green, 227; blue, 255 }  ,draw opacity=1 ][fill={rgb, 255:red, 226; green, 227; blue, 255 }  ,fill opacity=1 ] (250.45,93.21) .. controls (250.45,92.01) and (251.41,91.03) .. (252.6,91.03) .. controls (253.79,91.03) and (254.75,92.01) .. (254.75,93.21) .. controls (254.75,94.41) and (253.79,95.38) .. (252.6,95.38) .. controls (251.41,95.38) and (250.45,94.41) .. (250.45,93.21) -- cycle ;
\draw  [color={rgb, 255:red, 226; green, 227; blue, 255 }  ,draw opacity=1 ][fill={rgb, 255:red, 226; green, 227; blue, 255 }  ,fill opacity=1 ] (250.46,93.1) -- (254.72,93.1) -- (254.72,95.71) -- (250.46,95.71) -- cycle ;
\draw  [color={rgb, 255:red, 131; green, 135; blue, 255 }  ,draw opacity=1 ][fill={rgb, 255:red, 131; green, 135; blue, 255 }  ,fill opacity=1 ] (251.05,90.14) .. controls (251.05,89.27) and (251.75,88.56) .. (252.61,88.56) .. controls (253.48,88.56) and (254.18,89.27) .. (254.18,90.14) .. controls (254.18,91.02) and (253.48,91.73) .. (252.61,91.73) .. controls (251.75,91.73) and (251.05,91.02) .. (251.05,90.14) -- cycle ;
\draw  [color={rgb, 255:red, 226; green, 227; blue, 255 }  ,draw opacity=1 ][fill={rgb, 255:red, 226; green, 227; blue, 255 }  ,fill opacity=1 ] (251.05,89.77) .. controls (251.05,88.9) and (251.75,88.19) .. (252.61,88.19) .. controls (253.48,88.19) and (254.18,88.9) .. (254.18,89.77) .. controls (254.18,90.65) and (253.48,91.35) .. (252.61,91.35) .. controls (251.75,91.35) and (251.05,90.65) .. (251.05,89.77) -- cycle ;

\draw  [color={rgb, 255:red, 226; green, 227; blue, 255 }  ,draw opacity=1 ][fill={rgb, 255:red, 226; green, 227; blue, 255 }  ,fill opacity=1 ] (205.8,93.28) .. controls (205.8,92.08) and (206.76,91.11) .. (207.95,91.11) .. controls (209.14,91.11) and (210.1,92.08) .. (210.1,93.28) .. controls (210.1,94.49) and (209.14,95.46) .. (207.95,95.46) .. controls (206.76,95.46) and (205.8,94.49) .. (205.8,93.28) -- cycle ;
\draw  [color={rgb, 255:red, 226; green, 227; blue, 255 }  ,draw opacity=1 ][fill={rgb, 255:red, 226; green, 227; blue, 255 }  ,fill opacity=1 ] (205.81,93.17) -- (210.07,93.17) -- (210.07,95.78) -- (205.81,95.78) -- cycle ;
\draw  [color={rgb, 255:red, 131; green, 135; blue, 255 }  ,draw opacity=1 ][fill={rgb, 255:red, 131; green, 135; blue, 255 }  ,fill opacity=1 ] (206.4,90.22) .. controls (206.4,89.34) and (207.1,88.64) .. (207.96,88.64) .. controls (208.83,88.64) and (209.53,89.34) .. (209.53,90.22) .. controls (209.53,91.09) and (208.83,91.8) .. (207.96,91.8) .. controls (207.1,91.8) and (206.4,91.09) .. (206.4,90.22) -- cycle ;
\draw  [color={rgb, 255:red, 226; green, 227; blue, 255 }  ,draw opacity=1 ][fill={rgb, 255:red, 226; green, 227; blue, 255 }  ,fill opacity=1 ] (206.4,89.85) .. controls (206.4,88.97) and (207.1,88.26) .. (207.96,88.26) .. controls (208.83,88.26) and (209.53,88.97) .. (209.53,89.85) .. controls (209.53,90.72) and (208.83,91.43) .. (207.96,91.43) .. controls (207.1,91.43) and (206.4,90.72) .. (206.4,89.85) -- cycle ;

\draw  [color={rgb, 255:red, 226; green, 227; blue, 255 }  ,draw opacity=1 ][fill={rgb, 255:red, 226; green, 227; blue, 255 }  ,fill opacity=1 ] (191.07,70) .. controls (191.07,68.58) and (192.31,67.43) .. (193.85,67.43) .. controls (195.38,67.43) and (196.63,68.58) .. (196.63,70) .. controls (196.63,71.42) and (195.38,72.58) .. (193.85,72.58) .. controls (192.31,72.58) and (191.07,71.42) .. (191.07,70) -- cycle ;
\draw  [color={rgb, 255:red, 226; green, 227; blue, 255 }  ,draw opacity=1 ][fill={rgb, 255:red, 226; green, 227; blue, 255 }  ,fill opacity=1 ] (191.09,69.87) -- (196.58,69.87) -- (196.58,72.96) -- (191.09,72.96) -- cycle ;
\draw  [color={rgb, 255:red, 131; green, 135; blue, 255 }  ,draw opacity=1 ][fill={rgb, 255:red, 131; green, 135; blue, 255 }  ,fill opacity=1 ] (191.85,66.38) .. controls (191.85,65.34) and (192.75,64.5) .. (193.87,64.5) .. controls (194.99,64.5) and (195.89,65.34) .. (195.89,66.38) .. controls (195.89,67.41) and (194.99,68.25) .. (193.87,68.25) .. controls (192.75,68.25) and (191.85,67.41) .. (191.85,66.38) -- cycle ;
\draw  [color={rgb, 255:red, 226; green, 227; blue, 255 }  ,draw opacity=1 ][fill={rgb, 255:red, 226; green, 227; blue, 255 }  ,fill opacity=1 ] (191.85,65.94) .. controls (191.85,64.9) and (192.75,64.06) .. (193.87,64.06) .. controls (194.99,64.06) and (195.89,64.9) .. (195.89,65.94) .. controls (195.89,66.97) and (194.99,67.81) .. (193.87,67.81) .. controls (192.75,67.81) and (191.85,66.97) .. (191.85,65.94) -- cycle ;

\draw  [color={rgb, 255:red, 226; green, 227; blue, 255 }  ,draw opacity=1 ][fill={rgb, 255:red, 226; green, 227; blue, 255 }  ,fill opacity=1 ] (267.65,72.53) .. controls (267.65,71.69) and (268.41,71.01) .. (269.36,71.01) .. controls (270.31,71.01) and (271.08,71.69) .. (271.08,72.53) .. controls (271.08,73.37) and (270.31,74.06) .. (269.36,74.06) .. controls (268.41,74.06) and (267.65,73.37) .. (267.65,72.53) -- cycle ;
\draw  [color={rgb, 255:red, 226; green, 227; blue, 255 }  ,draw opacity=1 ][fill={rgb, 255:red, 226; green, 227; blue, 255 }  ,fill opacity=1 ] (267.66,72.45) -- (271.05,72.45) -- (271.05,74.28) -- (267.66,74.28) -- cycle ;
\draw  [color={rgb, 255:red, 131; green, 135; blue, 255 }  ,draw opacity=1 ][fill={rgb, 255:red, 131; green, 135; blue, 255 }  ,fill opacity=1 ] (268.13,70.38) .. controls (268.13,69.77) and (268.69,69.27) .. (269.37,69.27) .. controls (270.06,69.27) and (270.62,69.77) .. (270.62,70.38) .. controls (270.62,71) and (270.06,71.49) .. (269.37,71.49) .. controls (268.69,71.49) and (268.13,71) .. (268.13,70.38) -- cycle ;
\draw  [color={rgb, 255:red, 226; green, 227; blue, 255 }  ,draw opacity=1 ][fill={rgb, 255:red, 226; green, 227; blue, 255 }  ,fill opacity=1 ] (268.13,70.12) .. controls (268.13,69.51) and (268.69,69.01) .. (269.37,69.01) .. controls (270.06,69.01) and (270.62,69.51) .. (270.62,70.12) .. controls (270.62,70.73) and (270.06,71.23) .. (269.37,71.23) .. controls (268.69,71.23) and (268.13,70.73) .. (268.13,70.12) -- cycle ;

\draw  [color={rgb, 255:red, 226; green, 227; blue, 255 }  ,draw opacity=1 ][fill={rgb, 255:red, 226; green, 227; blue, 255 }  ,fill opacity=1 ] (247.76,63.31) .. controls (247.76,62.47) and (248.52,61.78) .. (249.47,61.78) .. controls (250.42,61.78) and (251.19,62.47) .. (251.19,63.31) .. controls (251.19,64.15) and (250.42,64.83) .. (249.47,64.83) .. controls (248.52,64.83) and (247.76,64.15) .. (247.76,63.31) -- cycle ;
\draw  [color={rgb, 255:red, 226; green, 227; blue, 255 }  ,draw opacity=1 ][fill={rgb, 255:red, 226; green, 227; blue, 255 }  ,fill opacity=1 ] (247.77,63.23) -- (251.16,63.23) -- (251.16,65.06) -- (247.77,65.06) -- cycle ;
\draw  [color={rgb, 255:red, 131; green, 135; blue, 255 }  ,draw opacity=1 ][fill={rgb, 255:red, 131; green, 135; blue, 255 }  ,fill opacity=1 ] (248.24,61.16) .. controls (248.24,60.55) and (248.8,60.05) .. (249.49,60.05) .. controls (250.17,60.05) and (250.73,60.55) .. (250.73,61.16) .. controls (250.73,61.77) and (250.17,62.27) .. (249.49,62.27) .. controls (248.8,62.27) and (248.24,61.77) .. (248.24,61.16) -- cycle ;
\draw  [color={rgb, 255:red, 226; green, 227; blue, 255 }  ,draw opacity=1 ][fill={rgb, 255:red, 226; green, 227; blue, 255 }  ,fill opacity=1 ] (248.24,60.9) .. controls (248.24,60.29) and (248.8,59.79) .. (249.49,59.79) .. controls (250.17,59.79) and (250.73,60.29) .. (250.73,60.9) .. controls (250.73,61.51) and (250.17,62.01) .. (249.49,62.01) .. controls (248.8,62.01) and (248.24,61.51) .. (248.24,60.9) -- cycle ;

\draw  [color={rgb, 255:red, 226; green, 227; blue, 255 }  ,draw opacity=1 ][fill={rgb, 255:red, 226; green, 227; blue, 255 }  ,fill opacity=1 ] (219.65,60.58) .. controls (219.65,59.87) and (220.32,59.29) .. (221.16,59.29) .. controls (221.99,59.29) and (222.67,59.87) .. (222.67,60.58) .. controls (222.67,61.29) and (221.99,61.87) .. (221.16,61.87) .. controls (220.32,61.87) and (219.65,61.29) .. (219.65,60.58) -- cycle ;
\draw  [color={rgb, 255:red, 226; green, 227; blue, 255 }  ,draw opacity=1 ][fill={rgb, 255:red, 226; green, 227; blue, 255 }  ,fill opacity=1 ] (219.66,60.52) -- (222.64,60.52) -- (222.64,62.06) -- (219.66,62.06) -- cycle ;
\draw  [color={rgb, 255:red, 131; green, 135; blue, 255 }  ,draw opacity=1 ][fill={rgb, 255:red, 131; green, 135; blue, 255 }  ,fill opacity=1 ] (220.07,58.77) .. controls (220.07,58.25) and (220.56,57.83) .. (221.17,57.83) .. controls (221.78,57.83) and (222.27,58.25) .. (222.27,58.77) .. controls (222.27,59.29) and (221.78,59.7) .. (221.17,59.7) .. controls (220.56,59.7) and (220.07,59.29) .. (220.07,58.77) -- cycle ;
\draw  [color={rgb, 255:red, 226; green, 227; blue, 255 }  ,draw opacity=1 ][fill={rgb, 255:red, 226; green, 227; blue, 255 }  ,fill opacity=1 ] (220.07,58.55) .. controls (220.07,58.03) and (220.56,57.61) .. (221.17,57.61) .. controls (221.78,57.61) and (222.27,58.03) .. (222.27,58.55) .. controls (222.27,59.07) and (221.78,59.48) .. (221.17,59.48) .. controls (220.56,59.48) and (220.07,59.07) .. (220.07,58.55) -- cycle ;

\draw  [color={rgb, 255:red, 226; green, 227; blue, 255 }  ,draw opacity=1 ][fill={rgb, 255:red, 226; green, 227; blue, 255 }  ,fill opacity=1 ] (227.42,103.25) .. controls (227.42,102.64) and (227.98,102.15) .. (228.66,102.15) .. controls (229.34,102.15) and (229.89,102.64) .. (229.89,103.25) .. controls (229.89,103.85) and (229.34,104.34) .. (228.66,104.34) .. controls (227.98,104.34) and (227.42,103.85) .. (227.42,103.25) -- cycle ;
\draw  [color={rgb, 255:red, 226; green, 227; blue, 255 }  ,draw opacity=1 ][fill={rgb, 255:red, 226; green, 227; blue, 255 }  ,fill opacity=1 ] (227.43,103.19) -- (229.87,103.19) -- (229.87,104.5) -- (227.43,104.5) -- cycle ;
\draw  [color={rgb, 255:red, 131; green, 135; blue, 255 }  ,draw opacity=1 ][fill={rgb, 255:red, 131; green, 135; blue, 255 }  ,fill opacity=1 ] (227.77,101.71) .. controls (227.77,101.27) and (228.17,100.91) .. (228.67,100.91) .. controls (229.16,100.91) and (229.56,101.27) .. (229.56,101.71) .. controls (229.56,102.15) and (229.16,102.5) .. (228.67,102.5) .. controls (228.17,102.5) and (227.77,102.15) .. (227.77,101.71) -- cycle ;
\draw  [color={rgb, 255:red, 226; green, 227; blue, 255 }  ,draw opacity=1 ][fill={rgb, 255:red, 226; green, 227; blue, 255 }  ,fill opacity=1 ] (227.77,101.52) .. controls (227.77,101.08) and (228.17,100.72) .. (228.67,100.72) .. controls (229.16,100.72) and (229.56,101.08) .. (229.56,101.52) .. controls (229.56,101.96) and (229.16,102.31) .. (228.67,102.31) .. controls (228.17,102.31) and (227.77,101.96) .. (227.77,101.52) -- cycle ;

\draw (187.67,37.73) node [anchor=north west][inner sep=0.75pt]   [align=left] {\begin{minipage}[lt]{50.32pt}\setlength\topsep{0pt}
\begin{center}
{\scriptsize Social Network}
\end{center}

\end{minipage}};
\draw (179.67,138.07) node [anchor=north west][inner sep=0.75pt]   [align=left] {\begin{minipage}[lt]{58.65pt}\setlength\topsep{0pt}
\begin{center}
{\scriptsize Recommendation}\\[-3pt]{\scriptsize System}
\end{center}

\end{minipage}};
\draw (281.57,56.48) node [anchor=north west][inner sep=0.75pt]   [align=left] {\begin{minipage}[lt]{58.65pt}\setlength\topsep{0pt}
\begin{center}
\textit{{\scriptsize Users' opinions }}
\end{center}

\end{minipage}};
\draw (82.88,173.29) node [anchor=north west][inner sep=0.75pt]   [align=left] {\begin{minipage}[lt]{62.22pt}\setlength\topsep{0pt}
\begin{center}
\textit{{\scriptsize Recommendations}}
\end{center}

\end{minipage}};

\end{tikzpicture}
    \caption{Recommendation system and social network feedback loop.}
    
    \label{feedbackloop}
\end{figure}

\emph{Contribution:}
We build upon~\cite{JC-JL-GZ-YD-LM:18} and~\cite{WSR-JWP-PF:22} by incorporating and extending the FJ opinion dynamics to capture interactions between multiple users and the influence of the recommendation system itself as a feedback loop, see Fig.~\ref{feedbackloop}.
Our contribution is twofold: first, we extend the work in~\cite{WSR-JWP-PF:22} from one single user to multiple, interacting, users. We then propose both a model-free and an idealized model-based control approach for user engagement maximization, with the latter employing a Model Predictive Control (MPC) strategy that accounts for future opinion evolution, see \cite{rawlings2017model} for a comprehensive review on MPC. Although the model-based approach is not attainable in practice (as it relies on perfect knowledge about the state and the model dynamics), it serves as a useful benchmark for the model-free controller. We demonstrate the feasibility and stability of these control approaches and we show that the model-free approach compares favorably to the model-based one. A numerical example is provided to illustrate a scenario in which the optimal controller may lead to drastic opinion shifts, underscoring the complex interplay between recommendation strategies and social outcomes.


\emph{Outline:}
Section~\ref{problemformulation} mathematically formulates engagement maximization and the closed-loop interaction between social network and recommendation system. Section~\ref{results} formulates the model-free and model-based optimal control solutions to the engagement maximization problem. 
 Section~\ref{numericalsimulation} shows numerical simulations and presents a scenario in which the proposed approaches can drastically affect users' opinions. Section~\ref{conclusions} concludes the work.

\emph{Notation:}\label{mathematicalpreliminaries}
Given $x\in \mathbb{R}^n$,  ($A \in \mathbb{R}^{m \times n}$),  $p\in [1,\infty]$, the $\ell_p$ (induced) norm of $x$, ($A$)  is denoted as $\lVert x \rVert_p$, ($\lVert A \rVert_p$). Given a matrix $H\in \mathbb{R}^{n\times n}$, we let $\lVert \cdot \rVert_H$, indicate the $H$-weighted $\ell_2$-norm. We denote as $|\mathbf{x}|$ the vector whose $i$-th entry, $|\mathbf{x}|_i$, is the $i$-th component of $\mathbf{x}$ in modulus, $|\mathbf{x}_i|$. We let the symbols $\mathbf{1}_n$,$\mathbf{0}_n$ and $\mathbb{I}_n$ denote the  $n$-dimentional all ones vector, zero vector and identity matrix, respectively. 
A matrix $A\in \mathbb{R}^{m \times n}$ is \emph{nonnegative} (\emph{positive}) if all its entries are nonnegative (positive).
A nonnegative matrix is said to be \emph{row-(sub)stochastic} if all of its row-sums are (less or) equal than one. 
A directed weighted  graph is a triple $\mathcal{G}~=~(\mathcal{V},\mathcal{E},W)$ where $\mathcal{V}~=~\{v_1,\dots, v_n\}$ is the set of vertices, $\mathcal{E}\subseteq \mathcal{V} \times \mathcal{V}$ is the set of arcs and $W~\in~\mathbb{R}^{n\times n}$ is the adjacency matrix of $\mathcal{G}$. An edge $(i,j)$ belongs to $\mathcal{E}$ if and only if $w_{ij}\neq 0$, where $w_{ij}$ denotes the $(i,j)$th entry of $W$.
A directed path in $\mathcal{G}$ is an ordered sequence of vertices such that any
ordered pair of consecutive vertices in the sequence is an edge in $\mathcal{G}$. The notation $x_{k|t}$ represents a variable that is predicted at time $t$ for $k$-steps ahead.

\section{Problem Formulation}\label{problemformulation}
We seek to design a recommendation system with the goal of maximizing users' engagement. The recommendation system and the social network of users interact in closed-loop, as shown in Fig~\ref{feedbackloop}. We first describe the dynamics through which opinions are formed. Then, we describe the closed-loop system and the mechanism of interaction between the users and recommendation system.

\subsection{Opinion dynamics: Friedkin-Johnsen model}\label{model}
 Let us consider a social network of $n+1$ users, $n\in \mathbb{N}$, described by a directed weighted graph $\mathcal{G}$, where the edge weights represent the frequency of interactions among users.
We assume that there is one single \emph{polarizing} topic of discussion over the network so that the opinion of a user $i~\in~\{1,...,n+1\}$ about the topic at time $t~\geq~0$ can be represented by a scalar variable $o_i(t)~\in~[0,1]$, where $o_i =1,(0)$ represents an opinion in strong (dis)agreement with that topic. 
The opinions of all users at each time step can then be collected into a vector $o(t)~\in~[0,1]^{n+1}$. 
The users' opinions evolve according to the FJ model~\cite{NEF-ECJ:90}, namely:
\begin{equation}\label{FJ}
    o(t+1) = (\mathbb{I}_{n+1}-\Lambda)Wo(t)+\Lambda o(0), 
\end{equation}
where $W~\in~\mathbb{R}^{(n+1)\times (n+1)}$ is the row-stochastic adjacency matrix associated with the directed weighted graph of users $\mathcal{G}$ and $\Lambda=\textrm{diag}(\lambda_1,...,\lambda_n)\in\mathbb{R}^{(n+1)\times (n+1)},\;\lambda_i\in[0,1], \forall i\in\{1,\dots,n+1\}$ is a diagonal matrix containing the stubbornness coefficients $\lambda_i$ of each user. The stubbornness coefficient represents the attachment of each user to its initial opinion. In particular, $\lambda_i~=~1$ corresponds to a fully prejudiced agent anchored to its own initial opinion, while $\lambda_i~=~0$ corresponds to a fully susceptible agent whose opinion is entirely influenced by its neighbours.

\begin{definition}[Prejudice and P-dependence]
    A user $i$, $i~\in~\{1, \dots, n+1\}$, is \emph{prejudiced} if $\lambda_i~>~0$. A user $j$, $j~\in~\{1, \dots, n+1\}$, is \emph{P-dependent} if it is either prejudiced or influenced by a prejudiced node $i~\neq~j$, i.e.,  a directed path from $i$ to $j$ exists in the graph. 
\end{definition}
\begin{definition}[$\lambda$-Connectivity]
   The directed graph $\mathcal{G}$, upon which the dynamics \eqref{FJ} evolves, is said to be \emph{$\lambda$-connected} if all users in the graph are P-dependent.
\end{definition}

We state the necessary and sufficient condition under which the system \eqref{FJ} admits an asymptotically stable equilibrium point.

\begin{theorem}[Convergence guarantees~\cite{AVP-RT:17}]\label{convergencethm}
    Given the FJ model \eqref{FJ}, the state $o^*$, defined as
    \begin{align*}
    o^* = (\mathbb{I}_{n+1}-(\mathbb{I}_{n+1}-\Lambda)W)^{-1}\Lambda o(0) ,
\end{align*}
    is the only equilibrium point of \eqref{FJ} and is asymptotically stable if and only if the system is $\lambda$-connected.
\end{theorem}

In the remainder of this paper, we consider only $\lambda$-connected user networks. Note that more generalized convergence results can be derived for networks that are not $\lambda$-connected~\cite{AVP-RT:17}.

\subsection{The recommendation system}
We manipulate the FJ model \eqref{FJ} to account for the recommendation system's influence on users' opinions. 
To this end, we consider the recommendation system to be an artificial user in the users graph $\mathcal{G}$ whose opinion is the expressed position of the recommended item. For the sake of notation simplicity, and without loss of generality, we consider the case where the recommendation system is the state $o_{n+1}$. Since the opinion of the recommendation agent is exogenous to the opinion formation process, we remove the $n+1$-th row of the FJ opinion dynamics model and rearrange the dynamics as follows
\begin{equation*}
    \tilde{o}(t+1) = (\mathbb{I}_{n} - \tilde{\Lambda})(\tilde{W}\tilde{o}(t)+  \tilde{{w}}_{{\rm col }(n+1)}o_{n+1}(t)) + \tilde{\Lambda}\tilde{o}(0),
\end{equation*}
where the symbol $\tilde{\cdot}$ indicates the vector (matrix) deprived of the last entry (row and column), while $\tilde{w}_{{\rm col }(n+1)}$ is the last column of $W$, deprived of the last entry.
By denoting as $x(t) := \tilde{o}(t)$,  $A:= (\mathbb{I}_n-\tilde{\Lambda})\tilde{W}$, $B:= (\mathbb{I} - \tilde{\Lambda})\tilde{w}_{{\rm col }(n+1)}$, $u(t):=o_{n+1}(t)$, we rewrite the above equation as

\begin{equation}\label{dynamics}
    x(t+1) = Ax(t)+Bu(t) + \tilde{\Lambda} x(0).
\end{equation}
The system \eqref{dynamics} is now in the form of a standard linear time-invariant (LTI) system with an additive constant term $\tilde{\Lambda} x(0)$, where the recommendation system is the control input $u$. This means that the recommendation system is able to suggest an opinion in the range $[0,1]$ on a given topic per time instant to all the connected users, which will influence their opinions.
Recommendation systems of this form have many real-world applications, including TV news channels, newspapers, or social media platforms. For example, a TV news channel must select one news item per time instant, expressing a position on a given topic which is shown to all viewers of the channel. Note that even the users that do not directly watch the TV channel, can be nevertheless influenced by it via interaction with other users. For the case of social media, instead, recommender systems algorithms do not have the capacity to capture a distinct model for each and every user. Instead, people are grouped into pools based on preferences and demographics that share the same recommendation items.

We start by showing that any input $u\in [0,1]$ will ensure that the state $x$ remains in the range $[0,1]^n$ as long as $x(0)~\in~[0,1]^n$. Such a property is crucial in our scenario to ensure that system \eqref{dynamics} indeed reflects the users' opinions.
\begin{proposition}[Well-posedness]\label{invariance}
    Consider system \eqref{dynamics}, and let $x(0), ~x(t)\in [0,1]^n$. Then, for any input $u(t)\in[0,1]$, we have that $x(t+1)\in[0,1]^n$.
\end{proposition}
\begin{proof}
We rearrange \eqref{dynamics} in the following form:
    \begin{equation*}
        x(t+1) = \underbrace{\begin{bmatrix}
            A & B & \tilde{\Lambda}
        \end{bmatrix}}_{{P}}\underbrace{\begin{bmatrix}
            x(t)\\
            u(t)\\
            x(0)
        \end{bmatrix}}_{{y(t)}}
    \end{equation*}
    By assumption, at time $t\geq 0$ the elements of $x(t)$, $u(t)$, and $x(0)$ are all in the range $[0,1]$, therefore, we have $\|{y}(t)\|_{\infty} \leq 1$.
    Now, based on the definition of $A$ and $B$, and considering the $i\in\{1, \dots, n\}$-th row of both the matrices, we have that $\|A_i\|_{\infty} + \|B_i\|_{\infty}\leq 1-\lambda_i$. Since the matrix $\tilde{\Lambda}$ is a diagonal matrix where $\tilde{\Lambda}_i=\lambda_i$, one can show that
    $ \|A_i\|_{\infty} + \|B_i\|_{\infty} + \|\tilde{\Lambda}_i\|_\infty \leq 1$.
    Finally, by considering the entire (non-negative) matrix $P$ and the vector $y$, it is now possible to see that
$ \|{P}{y}\|_{\infty} \leq \|{P}\|_{\infty}\|{y}\|_{\infty}\leq 1$. Finally, all entities involved being non-negative also guarantees $x(t+1)=Py(t)\geq 0$.
\end{proof}

\subsection{Opinion manipulation}
We focus on recommendation systems that seek solely to maximize user engagement. However, other scenarios can be imagined, for example in the context of political propaganda or bad actors seeking to exploit the recommendation system to steer users to some specific opinion. This raises an important question: To what extent would a recommendation system be able to manipulate user opinions? 
To answer this question, we show the set of reachable steady states  for $t\rightarrow\infty$, and give an explicit expression for the bounds of the set as function of the system matrices and initial conditions. 

\begin{proposition}[Reachability bounds]\label{reachability}
    Assume that the system \eqref{dynamics} is $\lambda$-connected. Then, the set of reachable steady states of the system \eqref{dynamics} with the admissible input set $u(t)~\in~[0,1]$ is the set of states component-wise bounded by the lower bound $l~:=~(\mathbb{I}_{n}-A)^{-1}\tilde{\Lambda} x(0)$ and the upper bound $m~:=~(\mathbb{I}_{n}-A)^{-1}(B+\tilde{\Lambda} x(0))$. 
\end{proposition}
\begin{proof}
    The proof trivially follows by noting that any steady-state of \eqref{dynamics} must satisfy the following equation
    \begin{align*}
        x^* = (\mathbb{I}_{n}-A)^{-1}(Bu+\tilde{\Lambda} x(0)),
    \end{align*}
since every element of the vector $(\mathbb{I}_{n}-A)^{-1}B$ is non-negative,  an increase in $u$ implies an increase in every component of $(\mathbb{I}_{n}-A)^{-1}Bu$. Substituting $u\in\{0,1\}$ one retrieves the expressions for $l$ and $m$.
\end{proof}

These bounds are useful in assessing a given network of users' susceptibility to manipulation. Recall from the definitions of $A$ and $B$ that, as the recommendation system's influence on node $i$ increases, the corresponding value $(B)_i$ also increases and the corresponding row of $A$ (representing influence from other users) must decrease (due to row-stochasticity of $W$). 
These findings can subsequently be integrated into current graph topology optimization algorithms, such as the one presented in~\cite{CM-CM-CET:18:wwwc}. By doing so one can aim to minimize the potential influence and manipulation that a recommendation system can exert over its users.


\subsection{Desired objective}\label{userengagement}

Recommendation systems provide content to maximize users’ engagement with a platform.  This objective will be pursued under the following assumption.
\begin{assumption}[Accessible opinions]\label{as:opinions}
All users' opinions are accessible to the recommendation system, i.e., the state $x$ is fully-measurable.
\end{assumption}
This assumption is motivated by the fact that the users' interaction with the platform via clicks or comments allows for the users' opinions to be estimated.

Inspired by~\cite{WSR-JWP-PF:22}, we model the decision of a user to engage with the provided content as a Bernoulli random variable $p_{\textrm{eng}}(t)~\sim~\mathcal{B}(1-\theta(x(t),u(t)))$. The function $\theta:[0,1]^n \times [0,1]\mapsto [0,1]$  measures the likelihood of users with opinions $x(t)~\in~[0,1]^n$ to engage with a recommended item that expresses a position $u(t)\in [0,1]$. We assume that users are affected by \emph{confirmation bias} and so they prefer items that most closely align with their current opinions, i.e., we seek a  $\theta$ such that $p_{\textrm{eng}}~=~1$ when $ x~=~u{\bf 1}_{n}$, and $p_{\textrm{eng}} = 0$ when $| x- u{\bf 1}_n|~=~{\bf 1}_{n}$. For this reason, we propose the  function
\begin{equation}\label{theta1}
    \theta(x(t),u(t)) = \lVert x(t)-u(t)\mathbf{1}_{n}\lVert_2^2,
\end{equation}
as a metric for engagement maximization over the network. 

The objective of the recommender is to maximize users' engagement over an infinite time horizon
\begin{equation}\label{motivatingproblem}
    \max \sum_{t=0}^{\infty}\mathbb{E}\left[p_{\textrm{eng}}(t)\right] = \min \sum_{t=0}^{\infty} \theta(x(t),u(t)).
\end{equation}
\section{Proposed Approaches}\label{results}
We develop two different approaches for the recommendation system: a model-free approach that is purely based on the knowledge of the opinions $x$, and an idealized model-based approach that considers the ideal case where the opinion dynamics ($A$, $B$, and $\tilde{\Lambda}$) are fully accessible.

\subsection{Model-free approach}\label{sec:modelfree}
The recommendation system solves a static optimization problem and iteratively provides a position that matches the users' instantaneous opinions. At each time instant $t\geq 0$, the control input $u_{\rm MF}(t)$ (where $MF$ stands for Model-Free) is chosen as 
\begin{align}\label{modelfree}
        u_{\rm MF}(t)&=\argmin_{u\in [0,1]}  \theta (x(t),u) 
         =\frac{1}{n}  \sum_{i=1}^{n}x_i(t). 
\end{align}
with $\theta (x(t),u)$ as in \eqref{theta1}. The optimization function in \eqref{modelfree} is known as $\ell_2$-distance seminorm \cite{GDP-KDS-FB-MEV:21m}. It is convex and continuously differentiable and it can be shown that the optimal value is achieved by the average of the opinions. 

The following theorem formalizes the convergence properties of the system in \eqref{dynamics} in closed-loop with the input \eqref{modelfree}. 
\begin{theorem}[MF and FJ analogy]\label{th:mf} The system in \eqref{dynamics} in closed-loop with the MF controller \eqref{modelfree} follows an FJ-type dynamics, is asymptotically stable and converges to
    \begin{align}\label{eq:eqpointMF}
    x^*_{\rm MF} = \bigg(\mathbb{I}_{n}-A-\frac{1}{n}B\mathbf{1}^\top_n\bigg)^{-1} {\tilde{\Lambda}} { x}(0) = S_{\rm MF} x(0).
\end{align}

\end{theorem}

\begin{proof}
  By substituting the optimal input $u_{\rm MF}$ into \eqref{dynamics}, we have
\begin{align}
\label{new_fj}
x(t+1) &= A x(t) +\frac{1}{n} B  \sum_{i=1}^{n}x_i(t)+ {\tilde{\Lambda}} x(0) \notag\\
    &= ({ \mathbb{I}_{n} - \tilde{\Lambda}})F{ x}(t) + {\tilde{\Lambda}} { x}(0),
\end{align}
where $ F = \tilde{ W} + \frac{1}{n}\tilde{ w}_{{\rm col} (n+1)}{\bf  1}_{n}^\top$. Note that  $F$ is the adjacency matrix of a $\lambda$-connected graph. Moreover, the $i$-th row of $F$ corresponds to the $i$-th row of the  $\tilde{W}$ with each entry incremented by $\frac{1}{n}(\tilde{w}_{{\rm col}(n+1)})_i$ and hence it is row stochastic.
Therefore, the system is an FJ-type dynamics, and one can make use of Theorem~\ref{convergencethm} to ensure that the it asymptotically converges to the equilibrium point \eqref{eq:eqpointMF}.
\end{proof}

\subsection{Model-based approach}
We propose an idealized model-based (MB) approach, namely, a model predictive controller, that optimally achieves the desired objective by predicting the evolution of the system over a user-defined time horizon $T>0$, and by making use of the matrices $A$, $B$, and $\tilde{\Lambda}$, which are instead assumed to be unknown in the model-free case.






First, we compute the desired optimal steady-state as 
    \begin{align}
    \label{eq:best_steady_state}
        (x_{\rm MB}^*, u_{\rm MB}^*) = \argmin_{x,u} &~\theta(x,u) \\
        \text{subject to } & x = Ax + Bu + \tilde{\Lambda} x(0), \notag\\
        &u \in [0,1]. \notag
        \end{align}
Note that as the objective function is strongly convex and the constraints are linear, an optimal solution to \eqref{eq:best_steady_state} always exists, as stated in the following lemma.
\begin{lemma}[MB Steady state]
The optimization problem \eqref{eq:best_steady_state} admits a closed form solution
\begin{align} \label{eq:explicitMPC}
        x^*_{\rm MB}=\bigg(\mathbb{I}_n-A-\frac{1}{v^\top \mathbf{1}_n}Bv^\top \bigg)^{-1}\tilde{\Lambda}x(0)=S_{\rm MB} x(0),
\end{align}
where $v=(\mathbb{I}_n-A)^{-1}B-\mathbf{1}_n\in \mathbb{R}^{n\times 1}$.
\end{lemma}
\begin{proof}
    The expression \eqref{eq:explicitMPC} can be computed by applying the Karush-Kuhn-Tucker (KKT) conditions to \eqref{eq:best_steady_state}. The derivation is omitted due to space constraints.
\end{proof}

At time $t$, given the current state $x(t)$, and the optimal steady-state $x_{\rm MB}^*$, the following optimal control problem (OCP) is solved 
 \begin{subequations}\label{empc}
    \begin{align}
        V^*_t := \min_{{x}_{\cdot|t},{u}_{\cdot|t}}&\sum_{k=0}^{T-1}\theta({x}_{k|t},u_{k|t})\label{empc_cost}\\
        \text{subject to } & {x}_{k+1|t} = {Ax}_{k|t}+{B}u_{k|t}+{\tilde{\Lambda} x}(0),\notag \\
        &{x}_{0|t} = {x}(t), \quad {x}_{T|t} = {x}_{\rm MB}^*\notag\\
        &u_{k|t}\in [0,1], \ \forall k \in [0,T-1] \notag.
    \end{align}
\end{subequations}
The optimal state and input sequences resulting from \eqref{empc} are defined as $x^*_{\cdot|t}$ and $u^*_{\cdot|t}$, while the optimal cost function $V^*$ is obtained in correspondence of the input $u_{\rm MB}(t)$, defined as $u_{\rm MB}(t) = u^*_{0|t}$. In closed-loop, the system evolves as follows 
\begin{align*}
x(t+1) =  Ax(t)+B u_{\rm MB}(t)+\tilde{\Lambda} x(0).
\end{align*}
By optimizing in a receding horizon fashion, MPC enables the recommendation system to respond to evolving user preferences and system dynamics, leading to improved transient performance and more effective recommendations.

\begin{theorem}[MB Convergence]
    Let the initial state $x(0)\in [0,1]^{n}$ and the OCP \eqref{empc} be feasible at time $t=0$. Then, the OCP \eqref{empc} is recursively feasible for all $t\geq 0$, the state $x(t)\in [0,1]^{n}$ for all $t\geq0$, and the steady-state $x_{\rm MB}^*$ is asymptotically stabilized.
\end{theorem}

\begin{proof}
Recursive feasibility can be shown by considering the following candidate input trajectory:
\begin{align}\label{eq:candidate_input}
    \bar{u}_{k|t+1} &= u^*_{k+1|t},\quad 
\bar{u}_{N-1|t+1} = u_{\rm MB}^*.    
\end{align}
with $ k\in[0, \dots, N-2]$.
To show asymptotic stability, we reformulate the stage cost $\theta$ as follows
\begin{align*}
    \theta(x,u)\!\!=\!\!\|(x,u)-(x_{\rm MB}^*,u_{\rm MB}^*)\|_H^2\!\!+\!\!(x,u)^\top\!\!h\!\!-\!\!\|(x_{\rm MB}^*, u_{\rm MB}^*)\|^2_H, 
\end{align*}
where the matrix $H$ and the vector $h$ are defined as follows
\begin{align*}
    H := \begin{bmatrix}
        \mathbb{I}_{n} & -\mathbf{1}_{n}\\
        -\mathbf{1}_{n}^{\top} & n
    \end{bmatrix},\qquad
    h := 2H(x_{\rm MB}^*, u_{\rm MB}^*)^\top.
\end{align*}
Even though the matrix $H$ is positive semi-definite, its null-space corresponds to the vector $\mathbf{1}_{n+1}$, which is the ideal state and input pair but it is not reachable for $\lambda$-connected graphs ($\Lambda\neq 0$).
Based on the candidate input trajectory defined in \eqref{eq:candidate_input}, the candidate state trajectory is
\begin{align*}
        \bar{x}_{k|t+1} &= x^*_{k+1|t}, \, k\in[0, \dots, N-1], \quad \bar{x}_{N|t+1} = x_{\rm MB}^*.   \notag 
\end{align*}
We refer to $\bar{V}_{t+1}$ as the cost associated with the candidate state and input sequences, i.e., $\bar{V}_{t+1} := V(\bar{x}_{\cdot|t+1},\bar{u}_{\cdot|t+1})$.

Note that, by optimality, we have that $V^*_{t+1}\leq  \bar{V}_{t+1}$, therefore, comparing the cost at two consecutive time instants yields
\begin{align*}
    &V^*_{t+1} - V^*_t\leq \bar{V}_{t+1} - V^*_t  \\
    =&\!\!\sum_{k=0}^{T-1}\theta(\bar{x}_{k|t+1},\bar{u}_{k|t+1})-\sum_{k=0}^{T-1}\theta({x}^*_{k|t},u^*_{k|t})\\
    =&\!\!-\!\!\theta(x^*_{0|t}, u^*_{0|t}) + \theta(\bar{x}_{N-1|t+1}, \bar{u}_{N-1|t+1}) \\
    =&\!\!-\!\!\|(x^*_{0|t},u^*_{0|t})-(x_{\rm MB}^*,u_{\rm MB}^*)\|_H^2 - (x^*_{0|t},u^*_{0|t})^\top h \\
    &+ \|(x_{\rm MB}^*,u_{\rm MB}^*)-(x_{\rm MB}^*,u_{\rm MB}^*)\|_H^2 + (x_{\rm MB}^*,u_{\rm MB}^*)^\top h \\
    =&\!\!-\!\!\|(x^*_{0|t},\!u^*_{0|t})\!\!-\!\!(x_{\rm MB}^*,\!u_{\rm MB}^*)\!\|_H^2\!\!-\!\!(x^*_{0|t},\!u^*_{0|t})\!^\top\!\!h\!\!+\!\!2\!\|\!(x_{\rm MB}^*,\!u_{\rm MB}^*)\!\|^2_H \\
    =&\!\!-\!\!\|(x^*_{0|t},\!u^*_{0|t})\|^2_H\!\!-\!\!\|(x_{\rm MB}^*,\!u_{\rm MB}^*)\|_H^2\!\!+\!\!2(x^*_{0|t},\!u^*_{0|t})\!^{\top}\!\!H\cdot\\
    &\cdot(\!x_{\rm MB}^*,\!u_{\rm MB}^*\!) - 2(x^*_{0|t},u^*_{0|t})^{\top}H(x_{\rm MB},u_{\rm MB})\\
    &+\!2\|(x_{\rm MB}^*,u_{\rm MB}^*)\|^2_H\!\!=\!-\|(x^*_{0|t},u^*_{0|t})\|^2_H\!+\!\|(x_{\rm MB}^*,u_{\rm MB}^*)\|_H^2,
\end{align*}
which ensures that the system asymptotically converges to the desired steady state $x_{\rm MB}^*$.
\end{proof}

Note that the known linear dynamics \eqref{dynamics} and quadratic user engagement function \eqref{theta1} would suggest the use of a Linear Quadratic Regulator (LQR) instead of MPC. However, for the system's interpretation to remain valid, the input must be constrained to be within $u\in [0,1]$, which can only be guaranteed by saturating the input of the LQR, resulting in a sub-optimal solution.
\subsection{Comparison between MF and MB approaches}
\begin{lemma}[Steady-states]
    The steady-state opinions of both the MF and MB approach belong to the convex-hull of the initial opinions, namely $x^*_{\rm MF}, x^*_{\rm MB}\in {\rm co}(x(0))$.  
\end{lemma}
\begin{proof}
To show that $x^*_{\rm MB}\in {\rm co}(x(0))$, one can equivalently show that $S_{\rm MB}$ is row-stochastic. The proof of this fact is analogous to the proof of row stochasticity of $V$ in \cite[Theorem 11.2]{CA:23}.
The fact that $x^*_{\rm MF}\in {\rm co}(x(0))$ follows from Theorem \ref{th:mf} and \cite[Theorem 11.2]{CA:23}.
\end{proof}

We analyze under what conditions the MF and MB approaches achieve the same steady-state.
\begin{theorem}[Steady-state equivalence]\label{ssthm}
Consider the FJ model in \eqref{dynamics} in closed-loop with \eqref{modelfree} and \eqref{empc}. Then, for all $x(0) \in {\rm span}\{ G^{-1}B\}^\perp$, with $G=\mathbb{I}_n-A-\frac{1}{n}B\mathbf{1}_n^\top$, the MB and MF approaches lead opinions to the same steady state. 
\end{theorem}
\begin{proof}
We write $S_{\rm MF}$ in \eqref{eq:eqpointMF} as $G^{-1}\tilde{\Lambda}$ and $S_{\rm MB}$ in \eqref{eq:explicitMPC} as $(G+ B\alpha D^\top)^{-1}\tilde{\Lambda}$ where $\alpha=n(C^\top \mathbf{1}_n-n)~\in~\mathbb{R}$ with  $C=(I-A)^{-1}B\in \mathbb{R}^{n}$, and $D=  \mathbf{1}_n\mathbf{1}_n^\top C -n C~\in~\mathbb{R}^{ n}$. By applying the Woodbury matrix identity, one gets that
\begin{align*}
    (G+B\alpha D^\top)^{-1} = G^{-1}-\underbrace{G^{-1}B(\alpha^{-1}+DG^{-1}B)^{-1}DG^{-1}}_{=:K}
\end{align*}
and we want to analyze the null space of $K\tilde{\Lambda}$. We note that the matrix $K\tilde{\Lambda}$ has rank one with image $G^{-1}B$  (note that $\alpha^{-1}+DG^{-1}B$ is a scalar). Therefore for all  $x(0) \in {\rm span}\{ G^{-1}B\}^\perp$, the MF and MB approaches lead to the same steady state solutions.
\end{proof}
\begin{remark}[Consensus invariance]
    The consensus  vector $\mathbf{1}_n\in {\rm ker}(K\tilde{\Lambda})$. In fact, $G^{-1}BDG^{-1}\tilde{\Lambda}\mathbf{1}_n =G^{-1}BDS_{\rm MF}\mathbf{1}_n = G^{-1}BD\mathbf{1}_n= \mathbf{0}_n$ since $S_{\rm MF}$ is row-stochastic and $D\mathbf{1}_n =0$. 
\end{remark}
\section{Numerical Simulations}\label{numericalsimulation}
\subsection{Approaches comparison}
\begin{figure}[]
      \centering
      \includegraphics[scale=.094]{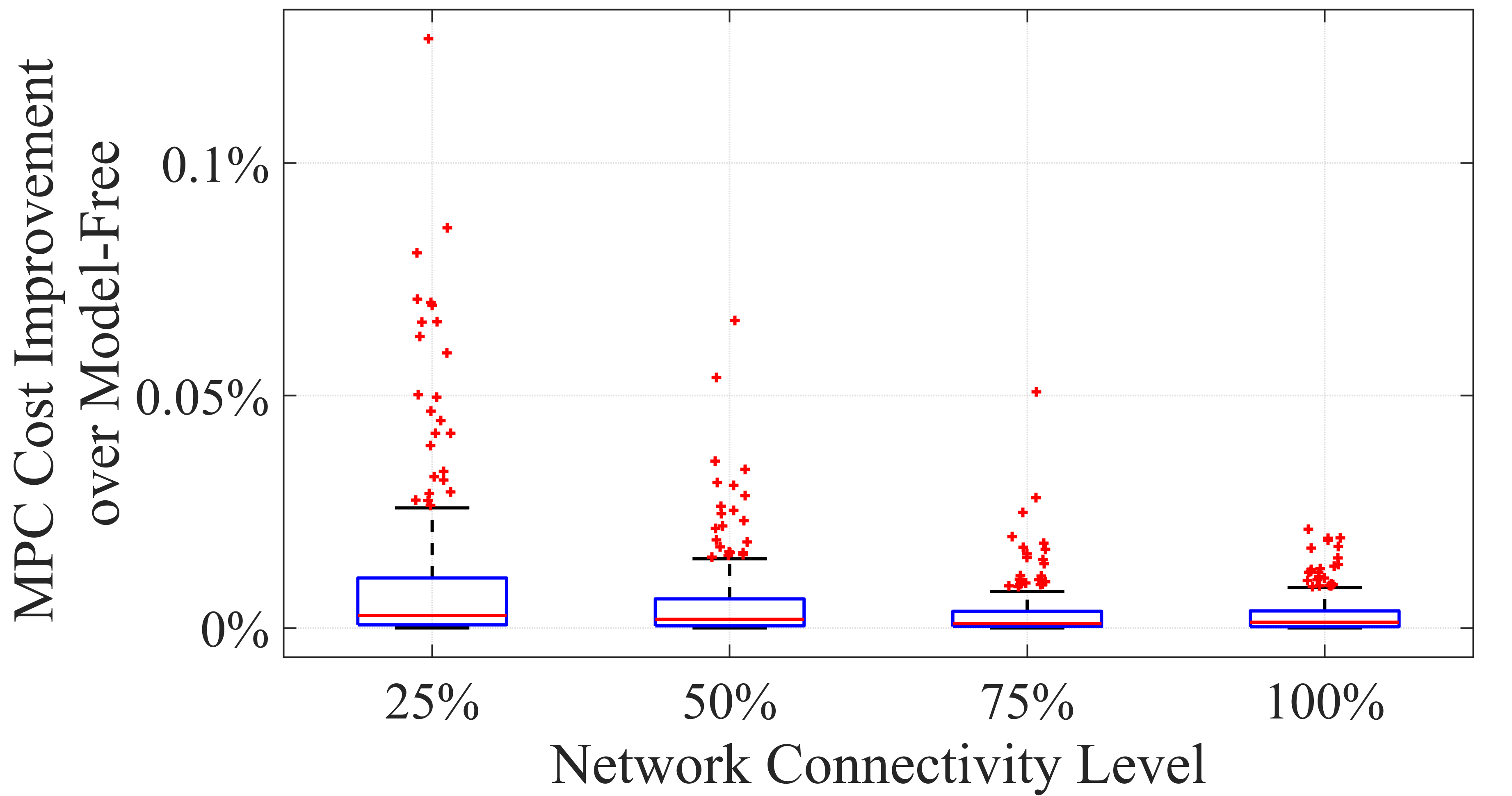}
      \caption{MPC cost improvement over Model-Free.}
      \label{boxplot}
\end{figure}
\begin{figure}[]
      \centering
      \includegraphics[scale=.094]{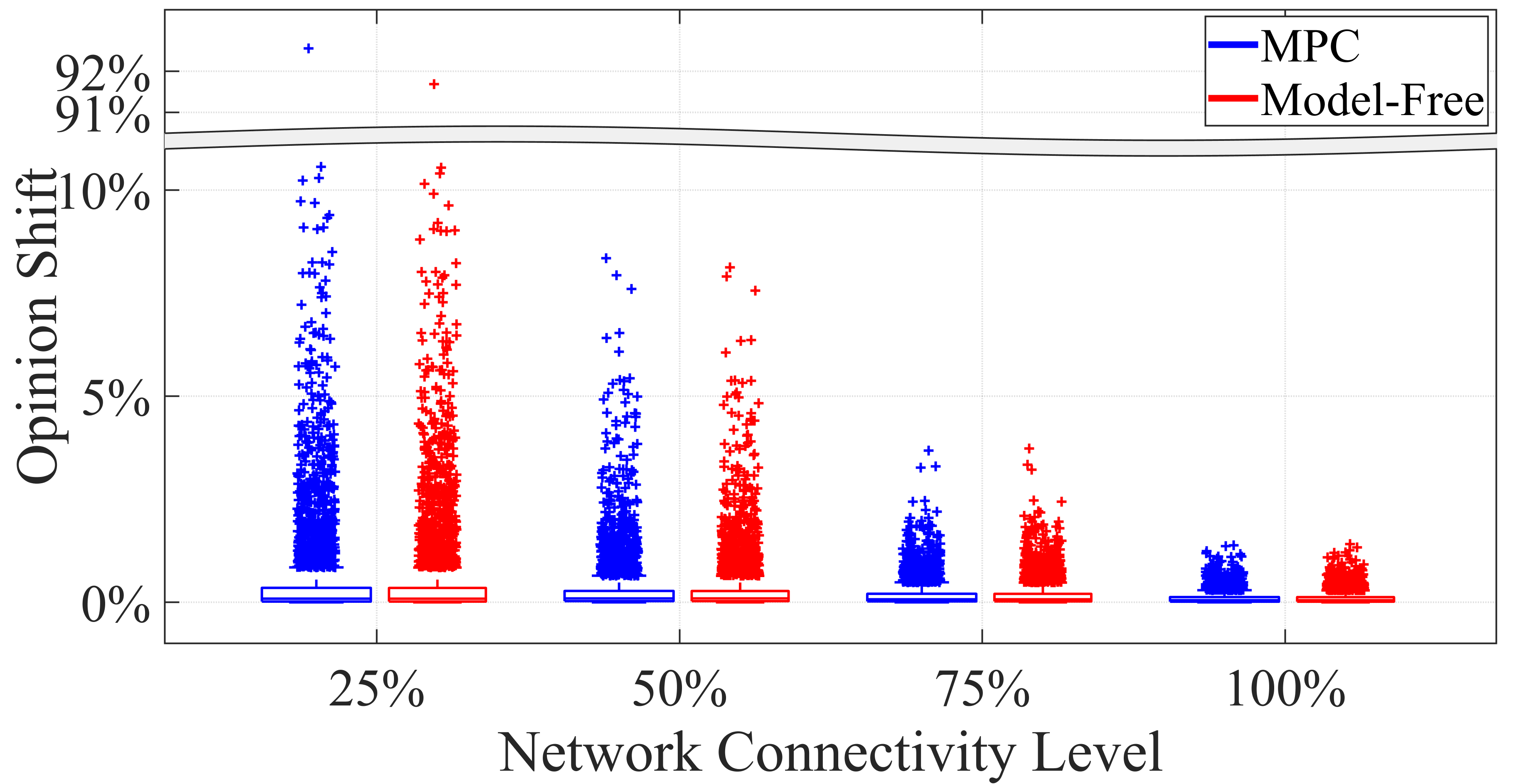}
      \caption{Opinion shift with MPC and Model-Free.}
      \label{opinionshift_boxplot}
\end{figure}
In order to compare the performance of the MB and MF recommendation systems, we performed a series of 1000 simulations for networks of 20 users.
Each simulation was initialized with random dynamics and initial opinions (all networks being $\lambda$-connected). The simulations were divided into four even subsets of network connectivity levels: 25\%, 50\%, 75\%, 100\%, representing the percentage of all possible inter-user connections that are filled. For the MB approach, the planning horizon was fixed to $T=50$. All simulations were performed for 50 time steps, by the end of which the systems reached a steady-state.
The MB outperformed the MF system across all trials, yet the benefits gained from the MB were marginal, as illustrated in Figure~\ref{boxplot}. This slight performance benefit is attributable to the MPC using knowledge of the opinion dynamics to steer users to a more similar opinion subspace, and is more pronounced in networks with lower connectivity. Given that determining the exact system dynamics matrices is virtually impossible for most user networks, it is plausible to anticipate inferior performance from the MB system compared to the MF approach in real-world settings.
Another question of interest is the effect of recommendation systems on users' steady-state opinions. To assess this, we simulated the user networks both with and without a recommendation system to serve as a baseline for comparison. We quantified the percentage opinion shift as $|x_{s} - x^{{\rm free}}_{s}| / x^{{\rm free}}_{s}$ (component-wise), which represents the change in a user's opinion relative to their steady-state opinion in free evolution. Our results in Fig. \ref{opinionshift_boxplot}, indicate that the shifts in opinion for both the MB and the MF strategies were similar across different trials, with the MB causing slightly greater shifts. 
For the sparsest network subset, the average user's opinion shift was a low 0.5\%, yet we observed a maximum shift of over 92.5\% for the MB and 91.6\% for the MF, highlighting the profound effect recommendation systems can have on shaping individual opinions. Notably, this change occurs regardless of any strategic foresight into the benefits of influencing user opinions. It is also important to note that as network connectivity increases, the opinion shift effects are greatly diminished. Given that online networks are often locally dense and globally sparse~\cite{NH-JF-MJM:07}, these findings highlight the importance of considering network structure in understanding and predicting the effects of recommendation systems on opinion dynamics.
\subsection{Fully prejudiced user}\label{sec:radical_user}
We now seek to highlight a scenario where the MPC-based system significantly outperforms the MF system, but also induces more severe shifts in steady-state user opinions.
The network shown in Fig.~\ref{radicalgraph} has one fully prejudiced \emph{radical user} (shown in grey) who has an extreme opinion,  $x_5(0)= 0$, and is completely biased, $\lambda_5 = 1$. Due to their extreme opinion, the radical user is socially isolated, i.e. other users do not consider their opinion. This scenario has a plausible interpretation: if an individual is identified as an extremist, moderate users are less likely to be influenced by him. 
\begin{figure}[]
    \centering
    \begin{tikzpicture}[scale=0.7, transform shape, vertex/.style={draw, circle, fill=white, inner sep=0pt, minimum size=14pt, font = \small}]
        \node[vertex] (4) at (0,0) {4};
        \node[vertex] (2) at (1.75,0) {2};
        \node[vertex] (3) at (3.5,2) {3};
        \node[vertex] (1) at (0,2) {1};
        \node[vertex, fill=gray] (5) at (5,0) {5};
        \node[vertex] (6) at (3.5,0) {6};
        \node[vertex, fill=pink] (7) at (1.75,2.75) {RS};
        
        \draw[->] (2) edge[bend right=10] node[pos=0.7, above right,font=\scriptsize] {0.041} (1);
        \draw[->] (4) edge[bend right=30] node[pos=.3, below right,font=\scriptsize] {0.397} (1);
        \draw[->] (7) edge[bend right=10] node[pos=0.35,above left,font=\scriptsize] {0.562} (1);
        
        \draw[->, out=-50, in =-130, loop, looseness=6, min distance=5mm] (2) to node[left=0.15cm, font=\scriptsize] {0.191} (2);
        \draw[->] (6) edge[] node[above,font=\scriptsize] {0.011} (2);
        \draw[->] (7) edge[] node[below right,font=\scriptsize] {0.798} (2);
        
        \draw[->] (6) edge[bend right=25] node[left,font=\scriptsize] {0.224} (3);
        \draw[->] (7) edge[bend left=10] node[pos=0.35,above right,font=\scriptsize] {0.776} (3);
        
        \draw[->] (1) edge[bend right=30] node[pos=0.5,above left,font=\scriptsize] {1.000} (4);
        
        \draw[->] (4) edge[bend right=55] node[right=0.55cm,font=\scriptsize] {1.000} (6);
        \draw[->] (3) edge[bend left=15] node[above right, font=\scriptsize] {0.472} (5);
        \draw[->] (6) edge[bend right=0] node[above, font=\scriptsize] {0.171} (5);
        \draw[->] (4) edge[bend right=55] node[below, font=\scriptsize] {0.357} (5);
        \node[right=0.2cm, font=\scriptsize, align=left] at (5) {\textit{radical user}};
        \node[above=0.2cm, font=\scriptsize, align=center] at (7) {\textit{recommender}}; 
    \end{tikzpicture}
    \caption{User network with radical user and recommendation system.}
    \label{radicalgraph}
\end{figure}
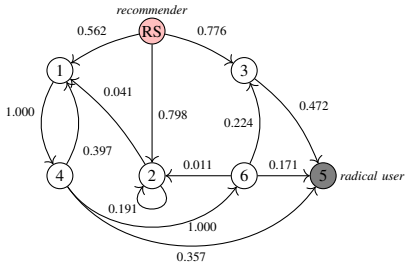

This scenario was simulated for the network in Fig.~\ref{radicalgraph}. The MPC approach used a planning horizon of $T=50$ steps, with both methods applied over a simulation duration of $50$ steps. User biases were set as $\tilde{\Lambda} = \textrm{diag}(0.011, 0.001, 0.092, 0.064, 1.000, 0.055)$ and initial conditions $x(0) = \begin{bmatrix}
    0.67 & 0.74 & 0.83 & 0.68 & 0 & 0.59
\end{bmatrix}^T$.  All users, other than the radical user, have very low biases. In this case, referring to Theorem \ref{ssthm}, the vector $K\tilde{\Lambda}x(0)$ assumes non-negligible values; consequently, we expect a consistent difference between the two approaches.
\begin{figure}[]
      \centering
      \includegraphics[scale=.1]{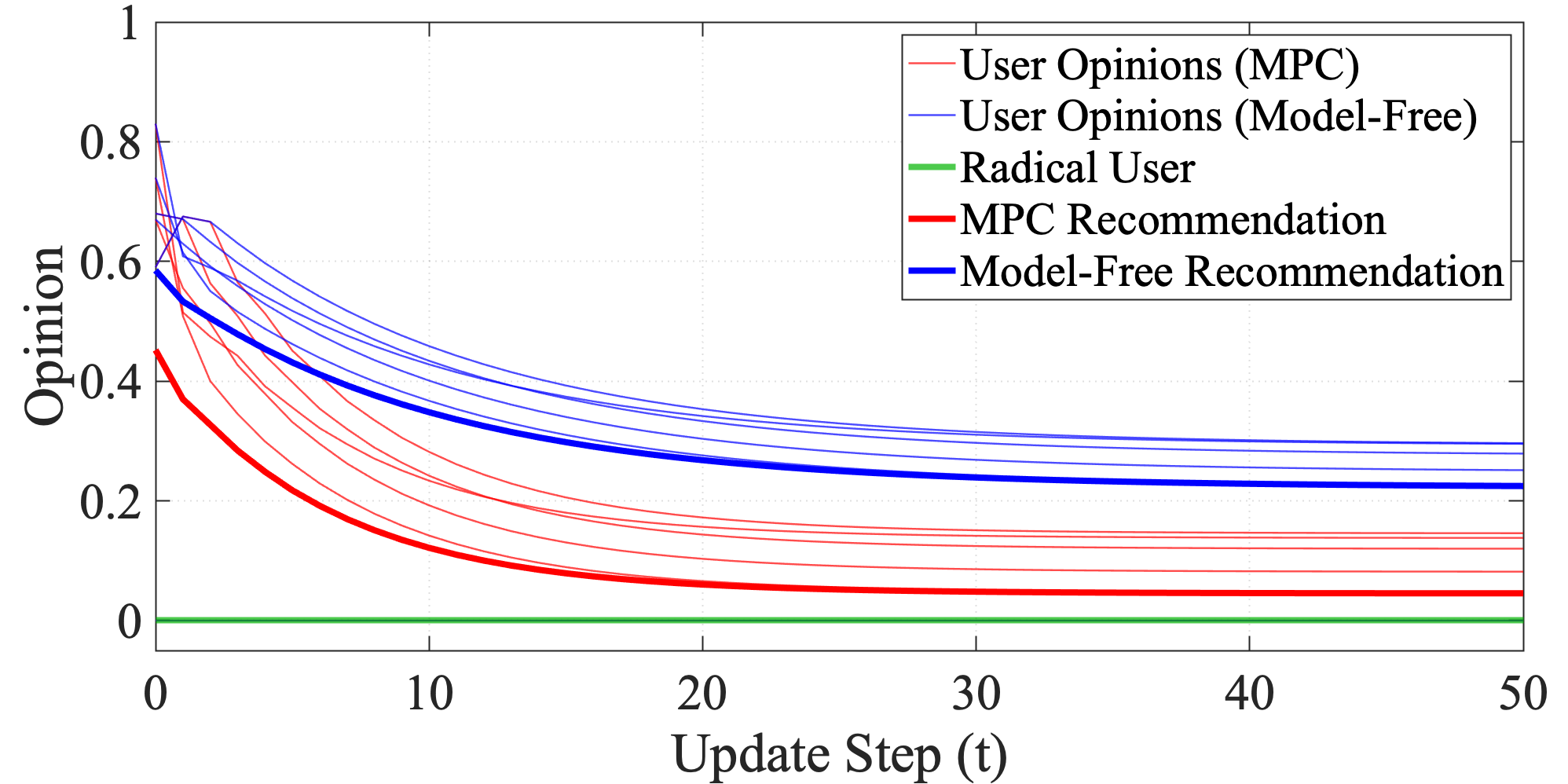}
      \caption{Opinion evolution for the network with a radical user.}
      \label{opinionevolution}
\end{figure}
The evolution of opinions under both recommendation system approaches is shown in Fig.~\ref{opinionevolution}. In particular, the MB system shows a 57\% improvement in steady-state cost over its MF counterpart, while also causing a 24.9\% greater average shift in opinions. These results can be intuitively understood by considering that, for the case of a single recommendation to all users, the recommendation system considers it optimal to provide recommendations aligned with the radical, stubborn user in order to shift other users towards the radical user's opinion. This enables the system to recommend content that is appealing to all users simultaneously in the long run. 
While both systems nudge opinions towards the radical viewpoint, 
the MB recommendation system's superior performance in engaging users stems from its use of the opinion dynamics model to control user opinions towards a more unified opinion subspace. However, this strategic advantage also increases the likelihood of unintentional radicalization of users in scenarios akin to the one assessed. 
The high degree of stubbornness of the radical user that gives rise to this phenomenon would be challenging to infer in real life settings: 
The performance gains of the MB approach are likely unachievable in practice.

\section{Conclusions}\label{conclusions}
In this paper, we modified the classical FJ model to propose a novel framework for recommendation systems as a closed-loop control problem. 
We proposed and compared a model-free and a model-based approach. Our analysis shows that for $\lambda$-connected graphs a model-free control approach performs virtually analogous to a model-based one. However, in some cases, for example if the graph has a stubborn radical user, then the MPC delivers better results in terms of final cost, potentially at the expense of desirable social outcomes. 
Future work may investigate extensions of the proposed approaches for different goals, such as radicalization mitigation, and exploit the notion of observability to infer connected users' opinions. 

\addtolength{\textheight}{-12cm}   




\bibliographystyle{unsrt}

\bibliography{alias,FB,references}

\end{document}